\documentclass[11pt, oneside]{article}   	
\usepackage{geometry}                		
\usepackage{graphicx}				
\usepackage{amssymb}
\usepackage{float}
\usepackage{amssymb}
\usepackage{amsmath}
\usepackage{amsthm}
\usepackage{color}
\usepackage{graphicx}
\usepackage{hyperref}
\usepackage{fullpage}
\usepackage{graphicx}
\usepackage{mathtools}
\usepackage{microtype}
\usepackage{bbm}
\usepackage{booktabs}

\usepackage[linesnumbered,algoruled,titlenotnumbered,vlined]{algorithm2e}
\DontPrintSemicolon

\renewenvironment{proof}{\vspace{-0.05in}\noindent{\bf Proof:}}%
        {\hspace*{\fill}$\Box$\par}
\newenvironment{proofof}[1]{\smallskip\noindent{\bf Proof of #1:}}%
        {\hspace*{\fill}$\Box$\par}
        {\hspace*{\fill}$\Box$\par}

\usepackage{algorithmic}

\usepackage{lmodern}
\usepackage[T1]{fontenc}
\usepackage[para,online,flushleft]{threeparttable}
\usepackage{multirow}

\newtheorem{theorem}{Theorem}
\newtheorem{lemma}[theorem]{Lemma}
\newtheorem{definition}[theorem]{Definition}

\newtheorem{proposition}[theorem]{Proposition}

\newcommand{\ie}{{\em i.e.}}

\title{Performance Analysis of Modified SRPT in Multiple-Processor Multitask Scheduling}
\author{Wenxin Li \\
Department of ECE\\
The Ohio State University\\
{\tt li.7328@osu.edu}\\
\and
Ness Shroff \\
Department of ECE and CSE\\
The Ohio State University\\
{\tt shroff.11@osu.edu}
}

\begin{document}
\maketitle
\begin{abstract}
In this paper we study the multiple-processor multitask scheduling problem in both deterministic and stochastic models, where each job have several tasks and is complete only when all its tasks are finished. We consider and analyze Modified Shortest Remaining Processing Time (M-SRPT) scheduling algorithm, a simple modification of SRPT, which always schedules jobs according to SRPT whenever possible, while processes tasks in an arbitrary order. The M-SRPT algorithm is proved to achieve a competitive ratio of $\Theta(\log \alpha +\beta)$ for minimizing response time, where $\alpha$ denotes the ratio between maximum job workload and minimum job workload, $\beta$ represents the ratio between maximum non-preemptive task workload and minimum job workload. In addition, the competitive ratio achieved is shown to be optimal (up to a constant factor), when there are constant number of machines. We further consider the problem under Poisson arrival and general workload distribution (\ie, M/GI/$N$ system), and show that M-SRPT achieves asymptotic optimal mean response time when the traffic intensity $\rho$ approaches $1$, if job size distribution has finite support. Beyond finite job workload, the asymptotic optimality of M-SRPT also holds for infinite job size distributions with certain probabilistic assumptions, for example, M/M/$N$ system with finite task workload. As a special case, we show that M-SRPT is asymptotic optimal in M/M/$1$ model, in which the task size distribution is allowed to have infinite support.


\end{abstract}
\section{Introduction}
With widespread applications in various manufacturing industries, scheduling jobs to minimize the total flow time (also known as response time, sojourn time and delay) is a fundamental problem in operation research that has been extensively studied. As an important metric measuring the quality of a scheduler, flow time, is formally defined as the difference between job completion time and releasing date, and characterizes the amount of time that the job spends in the system. 

Optimizing the flow time of single-task jobs has been considered both in offline and online scenarios. If preemption is allowed, the \emph{shortest remaining processing time} (SRPT) discipline is shown to be optimal in single machine environment. Many generalizations of this basic formulation become NP-hard, for example, minimizing the total flow time in non-preemptive single machine model and preemptive model with two machines~\cite{leonardi1997approximating}. When jobs arrive online, no information about jobs is known to the algorithm in advance, several algorithms with logarithmic competitive ratios are proposed in various settings~\cite{ azar2018improved,leonardi1997approximating}. On the other hand, while SRPT minimizes the mean response time sample-path wise, it requires the knowledge of remaining job service time. Gittins proved that the Gittins index policy minimizes the mean delay in an M/G/1 queue, which only requires the access to the information about job size distribution.

However, jobs with multiple tasks are more common and relevant in practice, which can take many different forms in modern computing environments. For example, for the objective of computing matrix vector product, we can divide matrix elements and vector elements into groups of columns and rows respectively, then the tasks correspond to the block-wise multiplication operations. Tasks can also be map, shuffle and reduce procedures in MapReduce framework. With the tremendous increasing in data size and job complexity, we cannot emphasize too much the importance of designing scheduling algorithms for jobs with multiple tasks. Though much progresses have been made in single-task job scheduling, there is a lack of theoretical understanding regarding \emph{multiple-processor multitask scheduling} (MPMS), where a job is considered to be completed only when all the tasks within the job are finished. A natural question that arises is, \emph{how to design an efficient scheduling algorithm to minimize the total amount time that the multitask jobs spend in the system}.

\paragraph{Related Work.}
There has been a large literature on single-task job scheduling, with parallel developments taking place in competitive analysis and queuing theory. However, little is known about multitask scheduling. Scully et. al~\cite{scully2017optimally} presented the first theoretical analysis of \emph{single-processor multitask scheduling} problem, and gave an optimal policy that is easy to compute for batch arrival, together with the assumption that the processing time of tasks satisfies the aged Pareto distributions. Sun et al.~\cite{sun2017near} studied the multitask scheduling problem when all the tasks are of unit size, and proved that among causal and non-preemptive policies, fewest unassigned tasks first (FUT) policy, earliest due date first (EDD) policy, and first come first serve (FCFS) are near delay-optimal in distribution (stochastic ordering) for minimizing the metric of average delay, maximum lateness and maximum delay respectively. To model the scenario when the scheduler has incomplete information about the job size, Scully et. al~\cite{scully2018optimal} introduced the multistage job model and proposed an optimal scheduling algorithm for multistage job scheduling in M/G/1 queue. In addition, the closed-form expression of the mean response time is given for the optimal scheduler. 

As an concrete example of multiple-processor multitask scheduling, there is a separate line of work focusing on the MapReduce framework. Here we only mention a few as examples. Wang et al.~\cite{wang2016maptask} studied the problem of scheduling map tasks with data locality, and proposed a map task scheduling algorithm consisting of the Join the Shortest Queue policy and MaxWeight policy. The algorithm asymptotically minimizes the number of backlogged tasks (which is directly related to the delay performance based on Little's law), when the arrival rate vector approaches the capacity region boundary. Zheng et al.~\cite{zheng2013new} proposed an online scheduler called available shortest remaining processing time (ASRPT), which is shown to achieve an efficiency ratio no more than two.

\paragraph{Contributions.}
In this paper, we investigate how to minimize the total response time of multitask jobs in a multi-server system and answer the aforementioned question. Our contributions are summarized as follows.

\begin{itemize}
\item We first propose Algorithm~\ref{schedulingalgo}, the Modified SRPT algorithm, for minimizing the total response time. Algorithm~\ref{schedulingalgo} is a simple modification of SRPT and achieves a competitive ratio of $O(\log \alpha+\beta)$, where $\alpha$ is the maximum-to-minimum job workload ratio, $\beta$ represents the ratio between maximum non-preemptive task workload and minimum job workload. It can be shown that no $o(\log \alpha+\beta)$-competitive algorithm exists when the number of machines is constant. In addition, $O(\log \alpha+\beta^{1-\varepsilon})$ is the best possible competitive ratio for the class of work-conserving algorithms.
%
\item Besides the worst case relative ratio above, we further prove our main result, absolute performance guarantees for Algorithm~\ref{schedulingalgo} under certain probabilistic structure on the input instances, in which the remaining workload bound established for the adversarial inputs contributes significantly to the stochastic analysis. Assuming that jobs arrive according to a Poisson process, \ie, in M/GI/$N$ system, we prove that the average response time incurred by Algorithm~\ref{schedulingalgo} is asymptotic optimal when load $\rho \rightarrow 1$, as long as the job size distribution has finite support. The assumption of finite job service time can be relaxed to finite task workload for exponentially distributed job size, \ie, M/M/$N$, together with other infinite distributions with certain properties on the tail of the distribution.
Last but not least, we prove the asymptotic optimality of Algorithm~\ref{schedulingalgo} in M/M/$1$ without the bounded task size assumption.
\end{itemize}



The remainder of this paper is organized as following. We introduce the problem definition, notations and necessary background in Section \ref{modelpre}. In Section \ref{deterministicalgo} we formally present Modified SRPT algorithm, together with the analysis of its competitive ratio and lower bounds. Section \ref{heavytraffic} is devoted to the proof of the asymptotic optimality of Modified SRPT in heavy traffic regime, together with the extensions to infinite job size distributions. We conclude our work in Section \ref{seccon}.

\section{Model and preliminaries}\label{modelpre}
\paragraph{Deterministic Model.}
We are given a set $\mathcal{J}=\{J_{1}, J_{2}, \ldots, J_{n}\}$ of $n$ jobs arriving online over time, together with a set of $N$ identical machines. Job $i$ consists of $n_{i}$ tasks and its workload $p_{i}$ is equal to the total summation of the processing time of tasks, \ie, $p_{i}=\sum_{\ell \in n_{i}}{p_{i, \ell}}$, where $p_{i, \ell}$ represents the processing time of the $\ell$-th task of job $i$. Tasks can be either preemptive or non-preemptive. A task is non-preemptive if it is not allowed to interrupt the task once it starts service, \ie, the task is run to completion. All the information of job $i$ is unknown to the algorithm until its releasing date $r_{i}$. Under any given scheduling algorithm, the completion time of job $j$ under the algorithm, denoted by $C_{j}$, is equal to the maximum completion time of individual tasks within the job. Formally, let $C^{(\ell)}_{j}$ be the completion time of task $\ell$ in job $j$, then $C_{j}=\max_{\ell\in [n_{i}]}{C^{(\ell)}_{j}}$. The response time of job $j$ is defined as $F_{j}=C_{j}-r_{j}$, our objective is to minimize the total response time $\sum_{j\in [n]}{F_{j}}$. 

Throughout the paper we use $\alpha=\max_{i\in [n]}{p_{i}}/\min_{i\in [n]}{p_{i}}$ to denote the ratio of the maximum to the minimum job workload. Let $\eta=\max\{p_{i, \ell}| \mbox{ task } \ell \mbox{ of job } i \mbox{ is non-preemptive}\}$ be the maximum processing time of a non-preemptive task, $\beta=\eta/\min_{i\in [n]}{p_{i}}$ be the ratio between $\eta$ and minimum job workload. In some sense, parameters $\beta$ and $\eta$ represent the degree of non-preemptivity and exhibits a trade-off between the preemptive and non-preemptive setting. More specifically, the problem approaches the preemptive case when $\eta$ is small, and degenerates to the non-preemptive case if all the jobs are consisted of a single non-preemptive task, in which $\eta$ reaches the maximum value of $\max_{i\in [n]}{p_{i}}$.

The definitions of work-conserving algorithms and competitive ratios are formally given as following.

\begin{definition}[Work-conserving scheduling algorithm] A scheduling algorithm $\pi$ is called work-conserving if it never idles machines when there exists at least one feasible job or task awaiting the execution in the system. Here a job or task is called feasible, if it satisfies all the given constraints of the system (e.g, preemptive and non-preemptive constraint, precedence constraint, etc).
\end{definition}

\begin{table}[H]
\centering
\begin{tabular}{rp{0.63\textwidth}}
\toprule
$N$ & number of machines\\
\midrule
$n$ &  number of jobs\\
\midrule
$r_{i}$ & arrival time of job $i$\\
\midrule
$p_{i}$& total workload of job $i$\\
\midrule
$\eta$ & maximum workload of a single non-preemptive task\\
\midrule
$\alpha$ & job size ratio: $\alpha=\max_{i\in [n]}{p_{i}}/\min_{i\in [n]}{p_{i}}$\\
\midrule
$\beta$ & relative ratio of the longest non-preemptive task: $\eta/\min_{i\in [n]}p_{i}$\\
\midrule
$\rho$ & traffic intensity $\rho=\mathbbm{E}[p_{i}]/(N\cdot \mathbbm{E}[\Delta r_{i}])$\\
\midrule
$\rho(y)$ & load composed of jobs with size $0$ to $y$: $\rho(y)=\lambda \cdot \int_{0}^{y}{tf(t)dt}$ \\
\midrule
$F^{\pi}_{\rho}$& job average response time under algorithm $\pi$ and load $\rho$  \\
\bottomrule
\end{tabular}
  \caption{Notation Table}
  \label{tb:notation}
\end{table}

\begin{definition}[Competitive ratio] The competitive ratio of online algorithm $\mathcal{A}$ refers to the worst ratio of the cost incurred by $\mathcal{A}$ and that of optimal offline algorithm $\mathcal{A}^{*}$ over all input instances $\omega$ in $\Omega$, \ie,
\begin{align*}
\mathcal{CR}_{\mathcal{A}}=\max_{\omega \in \Omega}	\frac{\mathrm{Cost}_{\mathcal{A}}(\omega)}{\mathrm{Cost}_{\mathcal{A^{*}}}(\omega)}.
\end{align*}
In the multiple-processor multitask scheduling problem, the cost is the total response time under instance $\omega=\{(r_{i}, \{p_{i, \ell}\}_{\ell \in [n_{i}]})\}_{i\in [n]}$.	
\end{definition}


%

\paragraph{Stochastic Model.}
In the stochastic setting, we assume that jobs arrive into the system according to a Poisson process with rate $\lambda$. Job processing times are i.i.d distributed with probability density function $f(\cdot)$. Formally,  we consider a sequence of M/GI/$N$ queues indexed by $n$, the traffic intensity of the $n$-th system is equal to $\rho^{(n)}=\lambda^{(n)}\cdot \mathbbm{E}[p^{(n)}_{i}]$, where $\lambda^{(n)}$ denotes the arrival rate of the $n$-th Poisson arrival process, job workload distribution has a density function of $f^{(n)}(\cdot)$. Stability of the queuing systems requires that $\rho^{(n)}<1$ for $\forall n$. As standard in the literature, we assume that $\rho^{(n)}\rightarrow1$ when $n\rightarrow \infty$. In this paper, we further assume that the probability density function $f^{(n)}(\cdot)$ is continuous. For notational convenience, we will suppress index $n$ whenever it is clear from the context. 

The stochastic analysis in this paper relies heavily on the concept of busy period, which is defined as following.

\begin{definition}[Busy Period~\cite{harchol2013performance}]
Busy period is defined to be the longest time interval in which no machines are idle. 
\end{definition}
\noindent We use $\mathsf{B}(w)$ to denote the length of a busy period with started by a workload of $w$.
It can be seen that $\mathsf{B}(\cdot)$ is an additive function~\cite{harchol2013performance}, \ie, $\mathsf{B}(w_{1}+w_{2})=\mathsf{B}(w_{1})+\mathsf{B}(w_{2})$ for $\forall w_{1}, w_{2}$, 
since a busy period with initial workload of $w_{1}+w_{2}$ can be regarded as a busy period started by initial workload $w_{2}$, following a busy period started by initial workload $w_{1}$. Moreover, for M/GI/$1$ queue, the length of a busy period with initial workload of $w$ and load $\rho$ is shown to be equal to $\mathsf{B}(w)=\mathbbm{E}[w]/(1-\rho)$ \cite{harchol2013performance}.

\subsection{Notations}
Notations of this paper are summarized in Table~\ref{tb:notation}. Most of our analysis are presented using asymptotic notations. We say $f(n)=o(g(n))$, $f(n)=O(g(n))$, $f(n)=\Theta(g(n))$, $f(n)=\Omega(g(n))$ if and only if $\limsup_{n \rightarrow \infty}{f(n)/g(n)}=0$, $\limsup_{n \rightarrow \infty}{f(n)/g(n)}<\infty$, $0<\liminf_{n\rightarrow \infty}{f(n)/g(n)}\leq \limsup_{n\rightarrow \infty}{f(n)/g(n)}<\infty$ and $\liminf_{n\rightarrow \infty}{f(n)/g(n)}>0$ respectively. All these notations only hide quantities that do not scale with $n$ (or $\rho^{(n)}$).




\section{Challenges with multi-task scheduling}

SRPT and its analysis do not easily generalize to multi-task scenario, due to the non-preemptivity of tasks. Firstly, in the analysis of single-task scheduling, the server only processes relevant work during the waiting time of the tagged job, under SRPT discipline. This holds for both single-server and multi-server settings. When jobs have multiple tasks, a challenge is raised: the system might be dealing with non-preemptive task of irrelevant jobs upon the arrival of the tagged job. It is unknown how the algorithm should be designed and how the amount of irrelevant workload involved can be bounded. Secondly, \cite{grosof2018srpt} bound the relevant work by comparing the multi-server SRPT system with single server system using SRPT. The analysis relies on the following fact: the workload difference is bounded in few-jobs interval and is non-decreasing in many-jobs interval, since the two systems are experiencing identical arrival sequence, while multi-server SRPT processes relevant workload at a maximum rate in such interval. However, in multi-task scheduling, the workload difference might be decreasing in many-jobs interval as resources might be used to process irrelevant jobs.

\section{Modified SRPT Algorithm and Competitive Ratio Analysis}\label{deterministicalgo}
The details of the Modified SRPT algorithm are specified in Algorithm \ref{schedulingalgo}. At each time slot $t$, jobs with non-preemptive task are kept processing on the machines, while the remaining machines are used to process jobs with smallest remaining workload. The main idea of Algorithm~\ref{schedulingalgo} is similar to SRPT, \ie, we utilize as many resources as possible on the job with smallest remaining workload, to reduce the number of alive jobs in a greedy manner, while satisfying the non-preemptive constraint. 

\begin{algorithm}[H]
\small
    \caption{Modified SRPT (M-SRPT)}
\label{schedulingalgo}
At time $t$, maintain the following quantities:\\
\begin{itemize}
\item For each job $i\in [n]$, maintain
\begin{itemize}
\item $W_{i}(t)$ \tcp{remaining workload}
\item $w_{i}(t)$ \tcp{remaining workload of the shortest single task being processed (if exists) or alive}
\end{itemize}
\item $\mathcal{J}_{1}(t)\leftarrow \{i\in [n]|w_{i}(t)=0\}$\tcp{Jobs with tasks that are finished at time $t$}
\item $\mathcal{J}_{2}(t)\leftarrow \{i\in [n]|\mbox{task of job $i$ is being processed at time $t$ and is preemptive}\}$
\end{itemize}
and assign alive jobs to the $|\mathcal{J}_{1}(t)\cup \mathcal{J}_{2}(t)|$ machines, where jobs with smaller value of $W_{i}(t)$ have a higher priority. 
\end{algorithm}

\subsection{Performance Analysis} 
Our main result is stated in the following theorem.
\begin{theorem}\label{algotheorem}
Algorithm~\ref{schedulingalgo} achieves a competitive ratio that is no more than 
\begin{align*}
\mathcal{CR}_{\mathrm{M-SRPT}}\leq 4\log \alpha+2\beta+8.
\end{align*}
\end{theorem}

To show the competitive ratio above, we divide the jobs into different classes and compare the remaining number of jobs under Algorithm~\ref{schedulingalgo} with that under optimal algorithm $\pi^{*}$. For any algorithm $\pi$, at time slot $t$, we divide the unfinished jobs into $\Theta(\log \alpha)$ classes $\{\mathcal{C}_{k}(\pi, t)\}_{k\in [\log \alpha +1]}$, based on their remaining workload. Jobs with remaining workload that is no more than $2^{k}$ and larger than $2^{k-1}$ are assigned to the $k$-th class. Formally, 
\begin{align*}
\mathcal{C}_{k}(\pi,t)=\Big\{i\in [n]\;\Big|\;W_i(\pi,t)\in (2^{k-1},2^{k}]\Big\},
\end{align*}
where $W_{i}(\pi,t)$ represents the unfinished workload of job $i$ at time $t$. In the following analysis, we use $\mathcal{C}^{[k]}(\pi, t)=\cup_{i=1}^{k}\mathcal{C}_{i}(\pi, t)$ to denote the collection of jobs in the first $k$ classes, and let $W_{\pi}^{[k]}(t)=\sum_{i=1}^{k} {W_{\pi}^{(i)}(t)}$ represent the total remaining workload of jobs in the first $k$ classes, where $W_{\pi}^{(k)}(\pi, t)$ denotes the amount of remaining workload of jobs in class $\mathcal{C}_k(\pi,t)$. $W_{\pi^{*}}^{(k)}(t)$ and $W_{\pi^{*}}^{[k]}(t)$ are defined in a similar way for the optimal scheduling algorithm $\pi^{*}$. 

Similar to the proof in \cite{leonardi2007approximating}, we first show the following lemma, which relates the remaining workload under M-SRPT with that under optimal algorithm $\pi^{*}$, then complete the proof of Theorem \ref{algotheorem} in Appendix \ref{appendixcomproof}.

\begin{lemma}\label{workloadlemma} For $\forall k, t\geq 0$, the unfinished workload under Algorithm~\ref{schedulingalgo} can be upper bounded as
\begin{align}\label{workloaddiff}
W_{\mathrm{M-SRPT}}^{[k]}(t)\leq W_{\pi^{*}}^{[k]}(t)+N \cdot (2^{k+1}+\eta).
\end{align}
\end{lemma}

\begin{proof}
In the following of the proof, we always divide jobs into different classes according to the remaining workload under M-SRPT, we suppress reference to M-SRPT in the notation of $\mathcal{C}_{k}$. Without loss of generality we can assume that $W_{\mathrm{M-SRPT}}^{[k]}(t)>W_{\pi^{*}}^{[k]}(t)$, otherwise Lemma~\ref{workloadlemma} already holds. Since the remaining workload under M-SRPT is strictly larger than that under the optimal algorithm, we claim that there must exist time in $(0, t]$, at which either
\begin{itemize}
\item Idle machines exist under M-SRPT;
\item Jobs with remaining workload (under M-SRPT) larger than $2^{k}$ are processed.
\end{itemize}
Otherwise, all the machines will be processing jobs belonging to set $\mathcal{C}_{[k]}(t)$ before time $t$, while no jobs in higher classes, \ie, $\cup_{i>k}\mathcal{C}_{i}(t)$, will be switched into class $\mathcal{C}_{[k]}(t)$. 
Combining with the fact that the initial workload under Algorithm~\ref{schedulingalgo} and optimal algorithm are identical, \ie, $W_{\mathrm{M-SRPT}}^{[k]}(0)=W_{\pi^{*}}^{[k]}(0)$, we can see that $W_{\mathrm{M-SRPT}}^{[k]}(t)$ should be no more than $W_{\pi^{*}}^{[k]}(t)$ and  the contradiction appears. 

Now consider the following two collections of time:
\begin{align*}
\mathcal{T}_{k}^{(1)}=&\Big\{\bar{t}\in[0,t] \Big|\mbox{ At time } \bar{t}, \mbox{ at least one machine is idle under Algorithm~\ref{schedulingalgo}} \Big\},\\
\mathcal{T}_{k}^{(2)}=&\Big\{\bar{t}\in[0,t]\Big|\mbox{ At time } \bar{t}, \mbox{ there exists } i>k \mbox{ such that at least one machine is}\\
&\;\;\;\;\;\;\;\;\;\;\;\;\;\;\;\;\;\;\mbox{  processing jobs in } \mathcal{C}_{i} \mbox{ under Algorithm~\ref{schedulingalgo}}\Big\}.
\end{align*}
Let $\bar{t}^{(i)}_{k}=\max\{t|t\in \mathcal{T}_{k}^{(i)}\}\;(i\in \{1,2\})$ be the last time slot in $\mathcal{T}_{k}^{(i)}$, based on which we divide our proof into the following two cases.

\paragraph{Case 1: $\bar{t}^{(1)}_k\geq \bar{t}^{(2)}_k$.} From the definition of $\bar{t}^{(1)}_k$, it can be seen that during $(\bar{t}^{(1)}_k, t]$, no machines are idle or process jobs with remaining workload larger than $2^{k}$ under Algorithm~\ref{schedulingalgo}, while the increment in remaining workload incurred by newly arriving jobs are identical for Algorithm~\ref{schedulingalgo} and $\pi^{*}$. In addition, it is important to point out that $([n]\setminus \mathcal{C}_{[k]}(\bar{t}^{(1)}_k)) \cap \mathcal{C}_{[k]}(\tilde{t})=\emptyset$ for $\forall \tilde{t}\in (\bar{t}^{(1)}_k, t]$, \ie, no job will switch from a higher class to $\mathcal{C}_{[k]}$ during $(\bar{t}^{(1)}_k, t]$. Hence 
\begin{align*}
W_{\mathrm{M-SRPT}}^{[k]}(t)-W_{\pi^{*}}^{[k]}(t)\leq W_{\mathrm{M-SRPT}}^{[k]}(\bar{t}^{(1)}_k)-W_{\pi^{*}}^{[k]}(\bar{t}^{(1)}_k).
\end{align*}
It suffices to prove the workload difference inequality (\ref{workloaddiff}) for $t=\bar{t}^{(1)}_{k}$, \ie,
\begin{align}\label{sufficetoshowineq}
W_{\mathrm{M-SRPT}}^{[k]}(\bar{t}^{(1)}_{k})\leq W_{\pi^{*}}^{[k]}(\bar{t}^{(1)}_{k})+N \cdot (2^{k+1}+\eta).
\end{align}
Note that there exists some idle machines at time $t=\bar{t}^{(1)}_k$, which implies that under Algorithm~\ref{schedulingalgo}, the number of jobs alive must be less than $N$. Hence $W_{\mathrm{M-SRPT}}^{[k]}(\bar{t}^{(1)}_k) \leq (N-1)\cdot 2^{k}$ and (\ref{sufficetoshowineq}) holds.

\paragraph{Case 2 : $\bar{t}^{(1)}_k < \bar{t}^{(2)}_{k}$.} According to the definition of $\bar{t}^{(2)}_k$, there exist jobs with remaining workload larger than $2^k$ being processed at $\bar{t}^{(2)}_{k}$, we use $\hat{\mathcal{J}}(\bar{t}^{(2)}_{k}) \subseteq [n]\backslash \mathcal{C}_{[k]}(\bar{t}^{(2)}_{k})$ to denote the collection of such jobs.

When all the tasks are processed preemptively, we can obtain (\ref{workloaddiff}) directly, as we are able to conclude that there are at most $N-1$ jobs in $\mathcal{C}_{[k]}(\bar{t}^{(2)}_k)$. This is because that tasks are allowed to be preempted, and Algorithm~\ref{schedulingalgo} selects a job with remaining workload larger than $2^{k}$ at time $\bar{t}^{(2)}_k$. Consequently $W_{\mathrm{M-SRPT}}^{[k]}(\bar{t}^{(2)}_k)\leq n_{\mathrm{M-SRPT}}^{[k]}(\bar{t}^{(2)}_k)\cdot 2^{k}$, and for $\forall t>\bar{t}^{(2)}_k$,
\begin{align*}
W_{\mathrm{M-SRPT}}^{[k]}(t)-W_{\pi^{*}}^{[k]}(t)\leq W_{\mathrm{M-SRPT}}^{[k]}(\bar{t}^{(2)}_k)-W_{\pi^{*}}^{[k]}(\bar{t}^{(2)}_k)+[N-n_{\mathrm{M-SRPT}}^{[k]}(\bar{t}^{(2)}_k)]\cdot 2^{k}\leq N\cdot 2^{k},
\end{align*}
where the first inequality follows from the fact that no more than $N-n_{\mathrm{M-SRPT}}^{[k]}(\bar{t}^{(2)}_k)$ jobs switches from higher classes to $\mathcal{C}_{[k]}(t)$, as there are at most $N-n_{\mathrm{M-SRPT}}^{[k]}(\bar{t}^{(2)}_k)$ jobs with remaining workload larger than $2^{k}$ are being processed at time $\bar{t}^{(2)}_k$. Hence Lemma~\ref{workloadlemma} holds.

Now for the case when there exist non-preemptive tasks, arguments above does not work, because machines may be processing tasks with remaining workload larger than $2^{k}$ and hence $n_{\mathrm{M-SRPT}}^{[k]}(\bar{t}^{(2)}_{k})$ may be larger than $N$.
Let $ r\in [N]$ be the number of tasks that are being processed at time $\bar{t}^{(2)}_{k}$ and belongs to $[n]\setminus \mathcal{C}_{[k]}(\bar{t}^{(2)}_k)$, and $t_s\leq \bar{t}^{(2)}_k $ be the latest starting processing time of these tasks. We divide our analysis into the following two subcases:
\begin{itemize}
\item \textbf{Case $2.1$:} No jobs switch from set $[n]\backslash \mathcal{C}_{[k]}(t_{s})$ to $\mathcal{C}_{[k]}(\bar{t}^{(2)}_{k})$ under Algorithm~\ref{schedulingalgo}. We use $\Delta_k$ to represent the increment of $W^{[k]}_{\mathrm{M-SRPT}}$, incurred by the newly arriving jobs during time period $[t_s, \bar{t}^{(2)}_k]$. Then we have:
\begin{align}
\label{ineq1}
	W^{[k]}_{\mathrm{M-SRPT}}(\bar{t}^{(2)}_k)-W^{[k]}_{\mathrm{M-SRPT}}(t_s)=-(N-r)(\bar{t}^{(2)}_k-t_s)+\Delta_k.
\end{align}
On the other hand, $W^{[k]}_{\pi^{*}}$, the remaining workload of jobs in class $\mathcal{C}_{[k]}$ under the optimal algorithm $\pi^{*}$, decreases at a speed that is no more than $N$ units of workload per time slot, hence
\begin{align}
\label{ineq2}
	W^{[k]}_{\pi^{*}}(\bar{t}^{(2)}_k)-W^{[k]}_{\pi^{*}}(t_s)\geq -N\cdot (\bar{t}^{(2)}_k-t_s)+\Delta_k.
\end{align}
According to the definition of $\bar{t}^{(2)}_k$, no jobs with remaining workload larger than $2^k$ are processed in $(\bar{t}^{(2)}_{k}, t]$. Compared with time $\bar{t}^{(2)}_k$, there are at most $r$ jobs switch from $[n]\backslash \mathcal{C}_{[k]}(\bar{t}^{(2)}_k)$ to set $\mathcal{C}_{[k]}(\bar{t}^{(2)+}_k)$. Therefore
\begin{align}
\label{ineq333}
	W^{[k]}_{\mathrm{M-SRPT}}(\bar{t}^{(2)+}_k)-W^{[k]}_{\pi^{*}}(\bar{t}^{(2)+}_k) \leq W^{[k]}_{\mathrm{M-SRPT}}(\bar{t}^{(2)}_k)-W^{[k]}_{\pi^{*}}(\bar{t}^{(2)}_k)+r \cdot 2^{k}.
\end{align}
Combining inequalities (\ref{ineq1})---(\ref{ineq333}), we can obtain
\begin{align*}
W^{[k]}_{\mathrm{M-SRPT}}(t)-W^{[k]}_{\pi^{*}}(t) & \leq W^{[k]}_{\mathrm{M-SRPT}}(\bar{t}^{(2)+}_k)-W^{[k]}_{\pi^{*}}(\bar{t}^{(2)+}_k)\\
& \leq W^{[k]}_{\mathrm{M-SRPT}}(t_s)-W^{[k]}_{\pi^{*}}(t_s)+r\cdot [2^k+(\bar{t}^{(2)}_k-t_s)]\\
&\leq (N-1)\cdot 2^k+r\cdot (\bar{t}^{(2)}_k-t_s)\\
&\leq N \cdot (2^k+\eta).
\end{align*}
The third inequality above holds since at time $t_{s}$, Algorithm~\ref{schedulingalgo} is required to do job selection and a job with remaining workload larger than $2^{k}$ is selected. The last inequality follows from the fact that $\bar{t}^{(2)}_k-t_s \leq \eta$, as $t_s$ is the starting time of a non-preemptive task that is still alive at time $\bar{t}_{k}^{(2)}$.

\item \textbf{Case $2.2$:} There exist jobs switching from set $[n]\backslash \mathcal{C}_{[k]}(t_{s})$ to $\mathcal{C}_{[k]}(\bar{t}^{(2)}_k)$ under Algorithm~\ref{schedulingalgo}. We use $\mathcal{J}_{s}$ to denote the collection of such switching jobs. It is essential to bound the number of switching jobs, which will incur an increment of $|\mathcal{J}_{s}|\cdot 2^{k}$ in the remaining workload of class $\mathcal{C}_{[k]}$. A straightforward bound is $|\mathcal{J}_{s}|\leq N\cdot (\bar{t}^{(2)}_{k}-t_{s})\leq N\cdot \eta$, since at most $N$ jobs receive service at each time slot, and hence the number of switching jobs is no more than $N$. However, this bound is indeed loose, we argue that 
\begin{align}
|\mathcal{J}_{s}|\leq N-r.
\end{align}
Notice that after a job switches to class $\mathcal{C}_{[k]}$ during $[t_{s}, \bar{t}^{(2)}_k]$, it will only be preempted by jobs that are also in class $\mathcal{C}_{[k]}$, which is due to the SRPT rule. According to the precondition of this case, there are $r$ jobs in set $[n]\setminus \mathcal{C}_{[k]}$ that are continuously being processed during $[t_{s}, \bar{t}_{k}^{(2)}]$, hence at most $N-r$ units of resources per time slot are available for the remaining jobs. Note that resources that are allocated to jobs in $\mathcal{C}_{[k]}$ will not be utilized for switching a job from a higher class to $\mathcal{C}_{[k]}$. In addition, finished jobs will have no contribution to the total remaining workload $W^{[k]}_{\mathrm{M-SRPT}}(t)$. Hence $|\mathcal{J}_{s}|$ is no more than $N-r$.

Furthermore, we can derive the following conclusion:
\begin{align*}
&W^{[k]}_{\mathrm{M-SRPT}}(t)-W^{[k]}_{\pi^{*}}(t)\\
 \leq &W^{[k]}_{\mathrm{M-SRPT}}(\bar{t}^{(2)+}_{k})-W^{[k]}_{\pi^{*}}(\bar{t}^{(2)+}_{k})\\
\leq & W^{[k]}_{\mathrm{M-SRPT}}(\bar{t}^{(2)}_{k})-W^{[k]}_{\pi^{*}}(\bar{t}^{(2)}_{k})+r \cdot 2^{k} +	N \tag{job switching at $t^{(2)}$}\\
\leq & [W^{[k]}_{\mathrm{M-SRPT}}(t_{s})-W^{[k]}_{\pi^{*}}(t_{s})+(N-r)\cdot 2^{k}\\
&+ N\cdot (\bar{t}^{(2)}_{k}-t_{s})]+ r \cdot 2^{k}+N\tag{job switching during $[t_{s}, \bar{t}^{(2)}_k]$}\\
\leq & N \cdot (2^{k+1}+\eta+1).\tag{$\bar{t}^{(2)}_{k}-t_{s} \leq \eta$}
\end{align*}
\end{itemize}
The proof is complete.
\end{proof}

\subsection{Competitive ratio lower bound}
The following lower bounds mainly follow from the observation that, multiple-processor multitask scheduling problem generalizes the single-task job scheduling problem in both preemptive and non-preemptive settings.
\begin{proposition}\label{lowerboundlemma1}
For multiple-processor multitask scheduling problem with constant number of machines, there exists no algorithm that achieves a competitive ratio of $o(\log \alpha+\beta)$.	
\end{proposition}
\begin{proof}
When $p_{\min}=\eta=1$, the problem degenerates to preemptive setting and no algorithm can achieve a competitive ratio of $o(\log \alpha)$ \cite{leonardi2007approximating}. When $\eta=p_{\max}$, the problem degenerates to the non-preemptive setting and $O(\beta)$ is the best possible competitive ratio if the number of machines is constant \cite{bunde2002approximating}. The proof is complete.
\end{proof}

\begin{proposition}\label{lowerboundlemma2}
For multiple-processor multitask scheduling problem, the competitive ratio of any work-conserving algorithms have an competitive ratio of $\Omega(\log \alpha+ \beta^{1-\varepsilon})$ for $\forall \varepsilon>0$.	
\end{proposition}
\begin{proof}
The reasoning is similar as the proof of Proposition~\ref{lowerboundlemma1}, since work-conserving algorithms cannot achieve a competitive ratio of $o(\beta^{1-\varepsilon})$ in the non-preemptive single-task job scheduling \cite{bunde2002approximating}.
\end{proof}

\section{Asymptotic Optimality of Modified SRPT with Poisson Arrival}\label{heavytraffic}
In this section we show that under mild probabilistic assumptions, Algorithm~\ref{schedulingalgo} is asymptotic optimal for minimizing the total response time in the heavy traffic regime. The result is formally stated as following.

\begin{theorem}\label{heavytrafficthm}
Let $F^{\mathrm{M-SRPT}}_{\rho}$ and $F^{\pi^{*}}_{\rho}$ be the response time incurred by Algorithm~\ref{schedulingalgo} and optimal algorithm respectively, when the traffic intensity is equal to $\rho$. In an M/GI/$N$ with finite job size distribution, Algorithm~\ref{schedulingalgo} is heavy traffic optimal, \ie,
\begin{align}\label{taskthreshold}
\lim\nolimits_{\rho \rightarrow 1}\frac{\mathbbm{E}[F^{\mathrm{M-SRPT}}_{\rho}]}{\mathbbm{E}[F^{\pi^{*}}_{\rho}]}=1.
\end{align}

\end{theorem}
%
The probabilistic assumptions here are with respect to the distribution of job size, \ie, the total workload of tasks. For the processing time of a single task, the only assumption we have is the upper bound $\eta$, which is finite since the job size distribution has finite support. It can be seen that the optimality result in~\cite{grosof2018srpt} corresponds to a special case of Theorem \ref{heavytrafficthm}.


\subsection{Average response time bound}\label{uplowbound}
We first remark that Lemma~\ref{workloadlemma} can be extended to any non-negative number $y\geq 0$.
\begin{lemma}\label{multaskabr} The difference of the amount of remaining workload under Algorithm~\ref{schedulingalgo} and that under $\mathrm{SRPT}$ algorithm in a single server system with speed $N$, is upper bounded by
\begin{align*}
\mathsf{W}^{\mathrm{M-SRPT}}_{\leq y}(t)-\mathsf{W}^{\mathrm{SRPT}_{1,N}}_{\leq y}(t)\leq N\cdot (2y+\eta), \forall y,t \geq 0,	
\end{align*}
where $\mathrm{SRPT}_{k,\ell}$ denotes the $\mathrm{SRPT}$ algorithm in a system with $k$ servers, and each server has a speed of $\ell$.
\end{lemma}

\begin{proof}
The proof is identical to that of Lemma~\ref{workloadlemma}.
\end{proof}

Our main goal is to derive the following analytical upper bound on $\mathbbm{E}[F_{\rho}^{\mathrm{M-SRPT}}]$. 
\begin{theorem}\label{lemma1}
The average response time under Algorithm~\ref{schedulingalgo} satisfies that
\begin{align}\label{unrefine1}
\mathbbm{E}{[F_{\rho}^{\mathrm{M-SRPT}}]}\leq \mathbbm{E}{[F_{\rho}^{\mathrm{SRPT}_{1,N}}]}+O\Big( \log \frac{1}{1-\rho}\Big).
\end{align}
\end{theorem}

\begin{proof}
Similar as the techniques in~\cite{grosof2018srpt,schrage1966queue}, we relate the response time of the tagged job with an appropriate busy period. 

Consider a tagged job with workload $x$, arriving time $r_{x}$ and completion time $C_{x}$. The computing resources of $N$ servers must be spent on the following types of job during $[r_{x}, C_{x}]$:
\begin{enumerate}
\item The system may be processing \emph{jobs with remaining workload larger than $x$, or some machines are idle, while the tagged job is in service}, because the number of jobs alive is smaller than $N$. We use $\mathsf{W}_{\mathrm{waste}}(r_{x})$ to represent the amount of such resources, then
\begin{align}\label{wastebound}
\mathsf{W}_{\mathrm{waste}}(r_{x})\leq (N-1)\cdot x,
\end{align}
which is indeed the same as Lemma $5.1$ in~\cite{grosof2018srpt}. The reason is straightforward---as the tagged job must be in service, hence the number of such time slots should not exceed $x$, and thus (\ref{wastebound}) holds.
\item The system may be dealing with \emph{jobs with remaining workload no more than $x$ at time} $r_{x}$, the amount of resources spent on this class is no more than $\mathsf{W}^{\mathrm{M-SRPT}}_{\leq x}(r_{x})$. Here for any algorithm $\pi$, we use $\mathsf{W}^{\pi}_{\leq x}(t)$ to denote the total workload of jobs with remaining workload no more than $x$ at time $t$. 
\item The system may be dealing with \emph{jobs that have a remaining workload larger than $x$ at time $t=r_{x}$, while the tagged job is not in service.} This is possible and happens only if the system is processing non-preemptive tasks, which belong to a job with total remaining workload larger than $x$. The tasks are in service before the arrival of the tagged job, and the non-preemptive rule allows the task to be served from time $r_{x}$ onwards.

Let $\mathsf{W}_{\mathrm{non-pm}}(r_{x})$ denote the total units of computing resources spent on this class of jobs during $[r_{x}, C_{x}]$. Our main argument for this class of jobs is,
\begin{align}\label{taskbound}
\mathsf{W}_{\mathrm{non-pm}}(r_{x})\leq (N^{2}+N)\cdot \eta+N\cdot x.
\end{align}
To see the correctness of inequality (\ref{taskbound}), we consider time intervals $[r_{x},r_{x}+\eta]$ and $(r_{x}+\eta, C_{x}]$ separately.
\begin{itemize}
\item Note that there are $N\cdot \eta$ computing resources during time $[r_{x}, r_{x}+\eta]$ in total, hence it is obvious to see that the amount of resources spent on this collection of jobs during $[r_{x}, r_{x}+\eta]$ cannot exceed $N\cdot \eta$.
\item We next show that in time interval $(r_{x}+\eta, C_{x}]$, the total amount of computing resources spent on such jobs is no more than $N^{2}\cdot \eta+N\cdot x$. Consider the following two types of jobs:
\begin{itemize}
\item Jobs that have a remaining workload larger than $x$ at time $t=r_{x}+\eta$. Note that jobs of this class will be processed after time $t=r_{x}+\eta$ only if the tagged job is in service, hence the amount of resources spending on such jobs are already taken into account in the first class above, \ie, the quantity $\mathsf{W}_{\mathrm{waste}}(r_{x})$, and we can ignore this subclass.
\item For the collection of jobs with remaining workload no more than $x$ at time $t=r_{x}+\eta$, it is clear to see that the remaining workload of such jobs at time $t=r_{x}$ must be no more than $x+N\cdot \eta$ (different tasks within the same job might be processed in parallel). Since there are at most $N$ such jobs in total, we can conclude that the remaining workload of jobs in this subclass must be no more than $N\cdot (x+N\cdot \eta)=N\cdot x+N^{2}\cdot \eta$, which implies that $\mathsf{W}_{\mathrm{non-pm}}(r_{x})\leq N\cdot x+N^{2}\cdot \eta+N\eta $ and (\ref{taskbound}) holds. 

\end{itemize}
\end{itemize}

\item \emph{Tagged job itself}. The amount of resources is equal to $x$, the size of the tagged job.
\item \emph{Newly arriving jobs during $[r_{x},C_{x}]$ with size no more than $x$}.
\end{enumerate}

Hence $\mathsf{T}^{\mathrm{M-SRPT}}_{x}$, the response time of the tagged job, is no more than the length of a busy period of a single server system with speed $N$, which starts at time $r_{x}$ and has a initial workload of
\begin{align*}
\mathsf{W}_{\mathrm{waste}}(r_{x})+\mathsf{W}_{\mathrm{non-pm}}(r_{x})+\mathsf{W}^{\mathrm{M-SRPT}}_{\leq x}(r_{x})+x.
\end{align*}
Combining with the aforementioned analysis, formally we have
\begin{align*}
\mathsf{T}^{\mathrm{M-SRPT}}_{x}&\leq_{st}\mathsf{B}^{(\rho_{x})}\Big(\mathsf{W}_{\mathrm{waste}}(r_{x})+\mathsf{W}_{\mathrm{non-pm}}(r_{x})+\mathsf{W}^{\mathrm{M-SRPT}}_{\leq x}(r_{x})+x\Big)\\
&\overset{(a)}{=} \mathsf{B}^{(\rho_{x})}\Big(\mathsf{W}_{\mathrm{waste}}(r_{x})+\mathsf{W}_{\mathrm{non-pm}}(r_{x})+x\Big)+ \mathsf{B}^{(\rho_{x})}\Big(\mathsf{W}^{\mathrm{M-SRPT}}_{\leq x}(r_{x})\Big)\\
& \overset{(b)}{\leq} \mathsf{B}^{(\rho_{x})}\Big(N^{2}\cdot \eta+(2x+\eta)\Big)+\mathsf{B}^{(\rho_{x})}\Big(\mathsf{W}^{\mathrm{M-SRPT}}_{\leq x}(r_{x})\Big)\\
&\overset{(c)}{\leq} \underbrace{\mathsf{B}^{(\rho_{x})}\Big(3N^{2}\cdot (\eta+x)\Big)}_{\Sigma_{1}}+\underbrace{\mathsf{B}^{(\rho_{x})}\Big(\mathsf{W}^{\mathrm{SRPT}_{1, N}}_{\leq x}(r_{x})\Big)}_{\Sigma_{2}},
\end{align*}
where $(a)$ follows from the additivity of busy period; In $(b)$ we utilize the upper bounds established in (\ref{wastebound}) and (\ref{taskbound}) and $(c)$ follows from Lemma~\ref{multaskabr}.

Note that the average response time under SRPT in a single server system is lower bounded as
\begin{align}
\mathbbm{E}{[F_{\rho}^{\mathrm{SRPT}_{1, N}}]}\geq  \mathbbm{E}_{x,r_{x}}{[\mathsf{B}^{(\rho(x))}(W^{\mathrm{SRPT}_{1, N}}_{\leq x}(r_{x}))]}= \mathbbm{E}_{x,r_{x}}{[\Sigma_{2}]},
\end{align}
where the first equality holds due to the Poission Arrivals See Time Average (PASTA) property~\cite{wolff1982poisson}. 
Note that
\begin{align}\label{ineqbound0}
\mathbbm{E}[\Sigma_{1}]&= O\Big(\mathbbm{E}\Big(\mathsf{B}^{(\rho(x))}(\eta+x)\Big)\Big)= O\Big(\mathbbm{E}\Big[\frac{\eta+x}{1-\rho(x)}\Big] \Big)\notag\\
&=O\Big(\log \frac{1}{1-\rho}\Big)+\mathbbm{E}[\eta]\cdot O\Big(\int_{0}^{\infty}\frac{{ f(x)}}{1-\rho(x)}dx \Big).
\end{align}
In addition,
\begin{align}\label{ineqbound}
\int_{0}^{\infty}\frac{f(x)}{1-\rho(x)}dx =\int_{0}^{\xi}\frac{f(x)}{1-\rho(x)}dx + \int_{\xi}^{\infty}\frac{f(x)}{1-\rho(x)}dx\leq \frac{1}{1-\rho(\xi)} +\frac{1}{\xi}\cdot \int_{\xi}^{\infty}\frac{xf(x)}{1-\rho(x)}dx,
\end{align}
where $\xi$ satisfies that $\rho(\xi) = \rho/2$. Note that
\begin{align*}
\rho(\xi)=\lambda \cdot \int_{0}^{\xi}{tf(t)dt}\leq \lambda \cdot\xi,
\end{align*}
hence we have $\xi \geq \mathbbm{E}[p_{i}]/2$. Then the right hand side of (\ref{ineqbound0}) can be further bounded as
\begin{align*}
\int_{0}^{\infty}\frac{f(x)}{1-\rho(x)}dx  \leq 2+ \frac{2}{\mathbbm{E}[p_{i}]} \cdot  \int_{0}^{\infty}\frac{xf(x)}{1-\rho(x)}dx= 2+ \frac{2}{\mathbbm{E}[p_{i}]} \cdot \log \frac{1}{1-\rho}.
\end{align*}
Therefore for any input instance, the average response time under Modified-SRPT, is no more than,
\begin{align}\label{unrefine}
\mathbbm{E}{[F_{\rho}^{\mathrm{M-SRPT}}]}&=\mathbbm{E}_{x,r_{x}}{[\mathsf{T}^{\mathrm{M-SRPT}}_{x}]}=\mathbbm{E}_{x,r_{x}}[\Sigma_{1}]+\mathbbm{E}_{x,r_{x}}[\Sigma_{2}]\notag\\
&\leq \mathbbm{E}{[F_{\rho}^{\mathrm{SRPT}_{1, N}}]}+O\Big( \log \frac{1}{1-\rho}\Big).
\end{align}
The proof is complete.
\end{proof}

\subsection{Existing lower bound for M/GI/1}

To start with, we consider the benchmark system consisting of a single machine with speed $N$, where all the tasks can be allowed to be served in preemptive fashion, \ie, the concept of task is indeed unnecessary in this setting. It is clear to see that the mean response time under optimal algorithm for this single machine system can be performed as a valid lower bound for the multitask problem, \ie, 
\begin{align}\label{ineqq1}
\mathbbm{E}{[F^{\pi^{*}}_{\rho}]}\geq \mathbbm{E}{[F^{\mathrm{SRPT}_{1,N}}_{\rho}]}.
\end{align}
It is well-known that SRPT minimizes the average response time in single server system.
For the case when job size distribution has finite support, Lin et al.~\cite{lin2011heavy} derived the heavy traffic growth rate of the average response time under SRPT~\cite{lin2011heavy}. 

\begin{lemma}[\cite{lin2011heavy}]\label{mg1bound}
In an $M/GI/1$ with finite job size distribution, the average response time under SRPT is in the order of
\begin{align*}
\mathbbm{E}[F^{\mathrm{SRPT}_{1,1}}_{\rho}]=\Theta\Big(\frac{1}{1-\rho}\Big).
\end{align*}
\end{lemma}


\subsection{Proof of optimality}
To achieve heavy traffic optimality, it suffices to show that the difference between the average response time under Algorithm~\ref{schedulingalgo} and the optimal algorithm is a lower order term, \ie,
\begin{align}\label{optcondi}
\lim_{n \rightarrow \infty}\frac{\mathbbm{E}{[F_{\rho}^{\mathrm{M-SRPT}}]}-\mathbbm{E}[F^{\mathrm{SRPT}_{1,N}}_{\rho}]}{\mathbbm{E}[F^{\mathrm{SRPT}_{1,N}}_{\rho}]}=0,
\end{align}
which holds according to Lemma~\ref{mg1bound} and inequality (\ref{unrefine})-(\ref{ineqq1}).

\subsection{Beyond Job Size Distribution with Finite Support }\label{secdis}
Up to this point, we have focused on job size distributions with finite support, which is rather restrictive. It is natural to consider various relaxations of this assumption.  In this section, we turn to other classes of job size distributions and the scenario when there are random number of tasks. These results provide complement to our developments about the theory of the asymptotic optimality of Modified SRPT.

\subsubsection{Exponential distribution and beyond}
\paragraph{M/M/$N$ model.} For the most elementary model of M/M/$N$, \ie, when the job service times are exponentially distributed, we have the following theorem, which only requires one additional assumption on task workload.
\begin{theorem}\label{mmonethm}
The average response time under Algorithm \ref{schedulingalgo} is asymptotic optimal in M/M/$N$, if task workload is finite.
\end{theorem}

\begin{proof}
The conclusion follows from the fact that in M/M/$1$ \cite{DBLP:journals/orl/Bansal05},
\begin{align*}
\frac{1/(18e)}{\mu(1-\rho)\log (1/(1-\rho))} \leq \mathbbm{E}[F^{\mathrm{SRPT-1}}_{\rho}]\leq  \frac{7}{\mu(1-\rho)\log (1/(1-\rho))}.
\end{align*}
\end{proof}

\paragraph{M/M/$1$ model.} For single server with Poisson arrival and 
exponentially distributed workload, we show that Modified SRPT is asymptotic optimal without any finite workload assumptions. 
\begin{theorem}\label{mm1opt}
Algorithm \ref{schedulingalgo} is asymptotic optimal in M/M/$1$.
\end{theorem}
We first introduce the following propositions that will be used in our proof.

\begin{proposition}\label{pro1}
The expected value of the maximum of $n$ i.i.d exponentially distributed random variables with mean $1/\mu$ is 
\begin{align*}
(1/\mu)\cdot \sum_{k=1}^{n}{(1/k)}=\Theta(\log n).
\end{align*}
\end{proposition}

\begin{proposition}[\cite{bansal2018achievable}]\label{pro2}
For M/M/$1$ model and any work-conserving algorithm, let $n_{\mathrm{busy}}$ be the number of arrivals in a busy period, then 
\begin{align*}
\mathbbm{E}[n_{\mathrm{busy}}]=O\Big(\frac{1}{1-\rho}\Big).   
\end{align*}
\end{proposition}

\begin{proofof}{Theorem \ref{mm1opt}}
From the proof of Theorem \ref{lemma1}, it can be verified that the average response time under Algorithm \ref{schedulingalgo} is no more than
\begin{align}\label{ineqeta}
\mathbbm{E}{[F_{\rho}^{\mathrm{M-SRPT}}]}\leq \mathbbm{E}{[F_{\rho}^{\mathrm{SRPT}_{1,N}}]}+\Big(1+\mathbbm{E}[\eta]\Big)\cdot O\Big( \log \frac{1}{1-\rho}\Big).    
\end{align}
Based on Proposition \ref{pro1}, Proposition \ref{pro2} and Jenson's inequality, we have
\begin{align*}
\mathbbm{E}[\eta] \leq (1/\mu)\cdot  \mathbbm{E}\Big[\sum_{k=1}^{n_{\mathrm{busy}}}{(1/k)}\Big]\leq (1/\mu)\cdot  \sum_{k=1}^{\mathbbm{E}[n_{\mathrm{busy}}]}{(1/k)} = O\Big(\log \frac{1}{1-\rho}\Big).  
\end{align*}
This implies that
\begin{align*}
\mathbbm{E}{[F_{\rho}^{\mathrm{M-SRPT}}]}- \mathbbm{E}{[F_{\rho}^{\mathrm{SRPT}_{1,N}}]} \leq O\Big(\log^{2} \frac{1}{1-\rho}\Big), 
\end{align*}
which is a lower order term. The proof is complete.
\end{proofof}

\paragraph{M/GI/$N$ model.} In addition to exponential distribution, Lin et al. \cite{lin2011heavy} also gave a characterization of the heavy-traffic behavior of SRPT with general job size distribution. We first introduce the background on Matuszewska index.

\begin{definition}[Upper Matuszewska Index \cite{lin2011heavy}]
Let $f$ be a positive function defined in $[0,\infty)$, the upper Matuszewska index is defined as the infimum of $\alpha$ for which there exists a constant $C=C(\alpha)$ such that for each $\bar{\lambda}>1$,
\begin{align*}
\lim_{x\rightarrow \infty}\frac{f(\gamma x)}{f(x)}\leq C\gamma^{\alpha},
\end{align*}
holds uniformly for $\lambda\in [1, \bar{\lambda}]$.
\end{definition}

\begin{proposition}[\cite{lin2011heavy}] \label{flowboundfact}In an $M/GI/1$ queue, if the upper Matuszewska index of the job size distribution is less than $-2$, then
\begin{align*}
\mathbbm{E}{[F^{\mathrm{SRPT-1}}_{\rho}]}=\Theta \Big(\frac{1}{(1-\rho)\cdot G^{-1}(\rho)}\Big),
\end{align*}
where $G^{-1}(\cdot)$ denotes the inverse of $G(x)=\rho_{\leq x}/\rho=\int_{0}^{x}{tf(t) dt}/\mathbbm{E}[p_{i}]$.
\end{proposition}
For example, exponential distribution has an upper Matuszewska index $M_{f}=-\infty$ and $G^{-1}(\rho)=\Theta(\log (1/(1-\rho)))$, hence Theorem \ref{mmonethm} is also implied by Proposition \ref{flowboundfact}. In addition, from Proposition~\ref{flowboundfact}, we can see that the following theorrem holds.
\begin{theorem}
The average response time under Algorithm \ref{schedulingalgo} is asymptotic optimal in M/GI/$N$, if task workload is finite, 
\begin{align*}
G^{-1}(\rho)= o\Big(\frac{1}{(1-\rho)\cdot \log(1/(1-\rho))}\Big), 
\end{align*}
and upper Matuszewska index of job size distribution is less than $-2$.
\end{theorem}
Examples include but not limited to \emph{Weibull distribution}, \emph{Pareto distribution} and \emph{regularly varying distributions}. Details are deferred in Appendix \ref{appendixdistributions}.

\subsubsection{Random number of tasks}

In the following proposition, we prove that the expected value of the maximum task size is finite, if the moment generating function of the task size distribution is finite.
\begin{proposition}\label{promax}
If the number of jobs and the number of tasks in each job are independently distributed with finite mean value, then the mean value of the maximum task size is no more than,
\begin{align*}
\mathbbm{E}[\eta]\leq \min_{s\in D(s)} \frac{\log(\mathbbm{E}[n_{t}])+\log m(s)}{s}<\infty,
\end{align*}
where $m(s)=\mathbbm{E}[e^{sp^{\ell}_{i}}]$ denotes the moment generating function of the task size distribution and $D(s)=\{s|m(s)<\infty\}$.
\end{proposition}
\begin{proof}
We first note that the expected value of the total number of tasks $\mathbbm{E}[n_{t}]=\mathbbm{E}[n]\cdot \mathbbm[n_{i}]<\infty$, given that $\mathbbm{E}[n], \mathbbm{E}[n_{i}]<\infty$. For any $s>0$, we have 
\begin{align*}
e^{s\mathbbm{E}[\eta]}=e^{\sum_{k\geq 0} s\mathbbm{E}[\eta|n_{t}=k]\cdot \mathbbm{P}(n_{t}=k)}=\prod_{k\geq 0}{e^{ s\mathbbm{E}[\eta|n_{t}=k]\cdot \mathbbm{P}(n_{t}=k)}},
\end{align*}
where $e^{ s\mathbbm{E}[\eta|n_{t}=k]}\leq \sum_{i=1}^{n_{t}}{\sum_{j=1}^{\ell}{\mathbbm{E}[e^{sp^{\ell}_{i}}]}}=k\cdot m(s)$, which implies that 
\begin{align*}
e^{s\mathbbm{E}[\eta]}\leq& \prod_{k\geq 0}{(k\cdot m(s))^{\mathbbm{P}(n_{t}=k)}}=m(s)\cdot e^{\sum_{k\geq 0}{\log k \cdot \mathbbm{P}(n_{t}=k)}}\\
=&m(s)\cdot e^{\mathbbm{E}[\log n_{t}]}\leq \mathbbm{E}[n_{t}]\cdot m(s),
\end{align*}
where the last inequality follows from the fact that $\mathbbm{E}[\log (n_{t})]\leq \log(\mathbbm{E}[n_{t}])$. Hence the expected maximum task size 
\begin{align*}
\mathbbm{E}[\eta]\leq \min_{s\in D(s)} \frac{\log(\mathbbm{E}[n_{t}])+\log m(s)}{s}<\infty.    
\end{align*}
\end{proof}

\begin{lemma}
In M/GI/$N$ queue, Algorithm~\ref{schedulingalgo} is heavy traffic optimal with random number of jobs and tasks, if the upper Matuszewska index of the job size distribution is less than $-2$ and $G^{-1}(\rho)=o(\frac{1}{(1-\rho)\cdot \log(1/(1-\rho))})$.
\end{lemma}
\begin{proof}
The Lemma mainly follows from Proposition \ref{flowboundfact}, Proposition \ref{promax} and inequality (\ref{ineqeta}). 
\end{proof}

\section{Conclusion}\label{seccon}
In this work, we study the multitask scheduling  problem, for which the optimal algorithms and tight analyses remain widely open for almost all settings. We propose Modified-SRPT algorithm, which achieves a competitive ratio that is order optimal when the number of machines is constant. Another appealing and more important property of Modified-SRPT is that, the average response time incurred under Poisson arrival is asymptotic optimal when the traffic intensity goes to $1$, if job service times are finite or exponentially distributed with finite task workload. We also show that this bounded workload assumption can be removed in M/M/$1$.

\bibliography{scheduling}
\bibliographystyle{plain}

\appendix

\section{Proof of Theorem \ref{algotheorem}}\label{appendixcomproof}
\begin{proof}
Let $n_{\mathrm{M-SRPT}}(t)$ and $n_{\pi^{*}}(t)$ represent the number of jobs alive at time $t$ under
Modified SRPT and optimal scheduler respectively. Without loss of generality, in the following of the proof we assume $\log p_{\max}$ and $\log p_{\min}$ are integers. For $\forall t\geq 0$, the number of unfinished jobs under the optimal algorithm is no less than,
\begin{align}
n_{\pi^{*}}(t) \ge & \sum_{k=\log p_{\min}}^{\log p_{\max}+1}{\frac{W_{\pi^{*}}^{(k)}(t)}{2^{k}}}\\
= &\sum_{k=\log p_{\min}}^{\log p_{\max}+1}{\frac{\Big[W_{\pi^{*}}^{[k]}(t)-W_{\pi^{*}}^{[k-1]}(t)\Big]}{2^{k}}}\tag{definition of $W^{[k]}_{\pi^{*}}(t)$}\\
=&\frac{W_{\mathrm{\pi^{*}}}^{[\log p_{\max}+1]}(t)}{2^{\log p_{\max}+1}}+\sum_{k=\log p_{\min}}^{\log p_{\max}+1}{\frac{W_{\mathrm{\pi^{*}}}^{[k]}(t)}{2^{k+1}}}\\
\geq& \sum_{k=\log p_{\min}}^{\log p_{\max}+1}{\frac{W_{\pi^{*}}^{[k]}(t)}{2^{k+1}}}\label{jobaliveoptimal}.
\end{align}
On the other hand, the number of jobs alive under Algorithm~\ref{schedulingalgo} can be upper bounded in a similar fashion,
\begin{align*}
n_{\mathrm{M-SRPT}}(t)\leq & \sum_{k=\log p_{\min}}^{\log p_{\max}+1}{\frac{W_{\mathrm{M-SRPT}}^{(k)}(t)}{2^{k-1}}}\\
=&\sum_{k=\log p_{\min}}^{\log p_{\max}+1}{\frac{\Big[W_{\mathrm{M-SRPT}}^{[k]}(t)-W_{\mathrm{M-SRPT}}^{[k-1]}(t)\Big]}{2^{k-1}}}\tag{definition of $W^{[k]}_{\mathrm{M-SRPT}}(t)$}\\
=&\sum_{k=\log p_{\min}}^{\log p_{\max}}{\frac{W_{\mathrm{M-SRPT}}^{[k]}(t)}{2^{k}}}+\frac{W_{\mathrm{M-SRPT}}^{[\log p_{\max}+1]}(t)}{2^{\log p_{\max}}}\\
\leq & \sum_{k=\log p_{\min}}^{\log p_{\max}+1}{\frac{W_{\mathrm{M-SRPT}}^{[k]}(t)}{2^{k-1}}}
\end{align*}
Using Lemma~\ref{workloadlemma}, we are able to relate the number of unfinished jobs under two algorithms,
\begin{align*}
n_{\mathrm{M-SRPT}}(t) \leq& \sum_{k=\log p_{\min}}^{\log p_{\max}+1}{\frac{W_{\mathrm{M-SRPT}}^{[k]}(t)}{2^{k-1}}}\\
\leq &\sum_{k=\log p_{\min}}^{\log p_{\max}+1}{\frac{W_{\pi^{*}}^{[k]}(t)}{2^{k-1}}}+\sum_{k=\log p_{\min}}^{\log p_{\max}+1}{\frac{N\cdot (2^{k}+\eta)}{2^{k-1}}}\\
\leq &4n_{\pi^{*}}(t)+N\cdot \Big(4\log \alpha+4\frac{\eta}{p_{\min}}+4\Big),
\end{align*}
where the last inequality follows from inequality (\ref{jobaliveoptimal}). To summarize, the competitive ratio of Algorithm~\ref{schedulingalgo} satisfies that
\begin{align*}
\mathcal{CR}_{\mathrm{M-SRPT}}=&\frac{\sum_{t: n_{\mathrm{M-SRPT}}(t)<N}{n_{\mathrm{M-SRPT}}(t)}+\sum_{t: n_{\mathrm{M-SRPT}}(t)\geq N}{n_{\mathrm{M-SRPT}}(t)}}{F^{\pi^{*}}}\\
\leq & \frac{\sum_{t: n_{\mathrm{M-SRPT}}(t)<N}{n_{\mathrm{M-SRPT}}(t)}}{F^{\pi^{*}}} + \frac{\sum_{t: n_{\mathrm{M-SRPT}}(t)\geq N}{4n_{\pi^{*}}(t)}}{F^{\pi^{*}}}\\
&+\Big(4\log \alpha+4\frac{\eta}{p_{\min}}+4\Big)\cdot \frac{\sum_{t: n_{\mathrm{M-SRPT}}(t)\geq N}{N}}{F^{\pi^{*}}}\\
\leq &4\log \alpha+4\frac{\eta}{p_{\min}}+8	,
\end{align*}
where the second inequality is due to Lemma~\ref{workloadlemma}. The proof is complete.
\end{proof}

\section{List of Distributions \cite{lin2011heavy}}\label{appendixdistributions}
\begin{itemize}
\item \emph{Weibull distribution.} Weibull distribution has a cumulative distribution function of $F(x)=1-e^{-\mu x^{\alpha}}$, upper Matuszewska index $M_{f}=-\infty$ and $G^{-1}(\rho)=\Theta({(\log (1/(1-\rho))}^{1/\alpha})$. Indeed exponential distribution is a special case of the Weilbull distribution with $\alpha=1$.
\item \emph{Pareto distribution.} A power-law job size distribution is often modeled with Pareto distribution, which has a cumulative distribution function of $F(x)=1-(x_{\min}/x)^{\alpha}\;(\alpha\geq 4)$ for $x\geq x_{\min}$. The upper Matuszewska index $M_{f}=\alpha$ and $\mathbbm{E}{[F^{\mathrm{SRPT-1}}_{\rho}]}=\Theta(1/(1-\rho)^{\frac{\alpha+2}{\alpha+1}})$. 
\item \emph{Regularly varying distributions.} More generally, the optimality condition also holds for regularly varying job size distribution $RV_{\alpha}$\;($\alpha\in (-\infty,-4)$) with cumulative distribution function $F(x)=1-L(x)\cdot x^{\alpha}$, where $L(\cdot)$ is a slowly varying function, \ie, $\lim_{x\rightarrow \infty}\frac{L(cx)}{L(x)}=1$ for any fixed $c>0$. The upper Matuszewska index of $RV_{\alpha}$ is equal to $\alpha$ and there exists a slowly varying function $L^{\prime}(\cdot)$ such that $G^{-1}(\rho)=L^{\prime}(1/(1-\rho))\cdot (1-\rho)^{1/(\alpha+1)}$, which implies that $\mathbbm{E}{[F^{\mathrm{SRPT-1}}_{\rho}]}=\Omega(1/(1-\rho)^{\frac{\alpha+2}{\alpha+1}-\epsilon})$.
\end{itemize}


\end{document}